\iffalse
\documentclass[runningheads]{llncs} \def\versy{conf}
\else
\documentclass[12pt]{llncsfp} \def\versy{full}
\fi
\def\fully{full}
\def\confy{conf}

\usepackage{amsmath,amssymb}
\usepackage[english,british]{babel}
\usepackage{enumerate}

\spnewtheorem{thm}{Theorem}{\bf }{\it }
\spnewtheorem{prop}[thm]{Proposition}{\bf }{\it }
\spnewtheorem{prob}[thm]{Open Problem}{\bf }{\it }
\spnewtheorem{cor}[thm]{Corollary}{\bf }{\it }
\spnewtheorem{lem}[thm]{Lemma}{\bf }{\it }
\spnewtheorem{defn}[thm]{Definition}{\bf }{\rm }
\spnewtheorem{rem}[thm]{Remark}{\bf }{\rm }
\spnewtheorem{exmp}[thm]{Example}{\bf }{\rm }
\spnewtheorem{clm}[thm]{Claim}{\bf }{\it }
\spnewtheorem{qust}[thm]{Question}{\bf }{\it }
\spnewtheorem{nota}[thm]{Notation}{\bf }{\rm }

\title{Ordered Semiautomatic Rings \\
  with Applications to Geometry\thanks{%
Frank~Stephan (PI) and Sanjay~Jain (Co-PI) are supported in part by
Singapore Ministry of Education Academic Research Funds Tier 2
MOE2016-T2-1-019 / R146-000-234-112 and MOE2019-T2-2-121 / R146-000-304-112.
Additionally, Sanjay~Jain was supported in part by
NUS grant C252-000-087-001. 
This project has also received funding from the European Union's Horizon 2020
research and innovation programme under the Marie Sk{\l}odowska-Curie
grant agreements No 794020 of Philipp Schlicht
(Project \emph{IMIC: Inner models and infinite computations}).
Ji Qi and Jacob Tarr worked on
this paper in the framework of the Undergraduate Research Opportunities
Programme at the NUS. The authors would like to thank Bakhadyr Khoussainov
and Sasha Rubin for correspondence.}}
\titlerunning{Ordered Semiautomatic Rings with Applications to Geometry}
\author{Ziyuan Gao\inst{1}, Sanjay Jain\inst{2},
Ji Qi \inst{1}, \newline Philipp Schlicht \inst{3},
Frank Stephan\inst{1}\inst{2}
and Jacob Tarr\inst{4}}
\authorrunning{Z.~Gao, S.~Jain, J.~Qi, P.~Schlicht, F.~Stephan and J.~Tarr}
\institute{Department of Mathematics, National University of Singapore\\
10 Lower Kent Ridge Road, S17, Singapore 119076, Republic of Singapore \\
fstephan@comp.nus.edu.sg
\and
Department of Computer Science, National University of Singapore\\
13 Computing Drive, COM1, Singapore 117417, Republic of Singapore \\
sanjay@comp.nus.edu.sg \and
School of Mathematics, University of Bristol \\
Fry Building, Woodland Road, Bristol, BS8 1UG, UK \\
philipp.schlicht@bristol.ac.uk \and
University of British Columbia\\ jacobdtarr@gmail.com}

\begin{document}

\maketitle

\begin{abstract}$\bf\!\!\!.\,$
The present work looks at semiautomatic rings with automatic addition
and comparisons which are dense subrings of the real numbers and asks
how these can be used to represent geometric objects such that certain
operations and transformations are automatic. The underlying ring has
always to be a countable dense subring of the real numbers and additions
and comparisons and multiplications with constants need to be automatic.
It is shown that the ring can be selected such that equilateral triangles
can be represented and rotations by $30^\circ$ are possible, while the
standard representation of the $b$-adic rationals does not allow this.
\end{abstract}

\section{Introduction}

\noindent
Hodgson \cite{Ho76,Ho83} as well as Khoussainov and Nerode \cite{KN95}
\ifx\versy\confy
and Blumensath and Gr\"adel \cite{BG00}
\fi
\ifx\versy\fully
and Blumensath and Gr\"adel \cite{Bl99,BG00}
\fi
initiated the study of automatic structures.
A structure $(A,\circ,\leq)$ of, say, an ordered semigroup
is then automatic iff there is an
isomorphic copy $(B,\circ,\leq)$ where $B$ is regular
and $\circ,\leq,=$ are automatic in the following sense:
A finite automaton reads all tuples of possible inputs and
outputs with the same speed in a synchronised way and accepts
these tuples which are valid tuples in the relations $\leq$ and $=$
or which are valid combinations $(x,y,z)$ with $x \circ y = z$ in
the case of the semigroup operation (function) $\circ$. For this,
one assumes that the inputs and outputs of relations and functions
are aligned with each other, like decimal numbers in addition,
and for this alignment -- which has to be the same for all operations --
one fills the gaps with a special character. So words are functions
with some domain $\{-m,-m+1,\ldots,n-1,n\}$ and some fixed range $\Sigma$
and the finite automaton reads, when processing a pair $(x,y)$ of inputs,
in each round the symbols $(x(k),y(k))$ where the special symbol $\# \notin
\Sigma$ replaces $x(k)$ or $y(k)$ in the case that these are not defined.
See Example~\ref{gridexample} below for an example of a finite automaton
checking whether $x+y=z$ for numbers $x,y,z$; here a finite automaton
computes a function by checking whether the output matches the inputs.
Automatic functions are characterised as those computed by
a position-faithful one-tape Turing machine in linear time \cite{CJSS12}.
\ifx\versy\fully
Position-faithful means that input and output start at the same position
on the Turing tape and the Turing machine overwrites the input by the
output. Furthermore, note that if a structure has several representatives
per element then a finite automaton must recognise the equality.
\fi

The reader should note, that after Hodgson's pioneering work \cite{Ho76,Ho83},
Epstein, Cannon, Holt, Levy, Paterson and Thurston \cite{ECHLPT92}
argued that in the above formalisation, 
automaticity is, at least from the viewpoint
of finitely generated groups, too restrictive. They furthermore wanted that
the representatives of the group elements are given as words over the
generators, leading to more meaningful representatives than arbitrary strings.
Their concept of automatic groups led, for finitely generated groups, to a
larger class of groups, though, by definition, of course 
it does not include groups which require infinitely many generators;
groups with infinitely many generators, to some extent, were covered
in the notion of automaticity by Hodgson, Khoussainov and Nerode. Nies,
Oliver and Thomas provide in several papers
\ifx\versy\confy
\cite{NT08,OT05}
\fi
\ifx\versy\fully
\cite{Ni07,NT08,OT05}
\fi
results which contrast and compare these two notions of automaticity and
give an overview on results for
groups which are automatic in the sense of Hodgson, Khoussainov and Nerode.
\ifx\versy\fully
Kharlampovich, Khoussainov and Miasnikov \cite{KKM11} generalised the notion
further to Cayley automatic groups. Here a finitely generated
group $(A,\circ)$ is Cayley
automatic iff the domain $A$ is a regular set, for every
group element there is a unique representative in $A$ and, for every
$a \in A$, the mapping $x \mapsto x \circ a$ is automatic.
\fi

Jain, Khoussainov, Stephan, Teng and Zou \cite{JKSTZ17} investigated the
general approach where, in a structure for some relations and functions,
it is only required that the versions of the functions or relations with
all but one variable fixed to constants is automatic. Here the
convention is to put the automatic domains, functions and relations before
a semicolon and the semiautomatic relations after the semicolon.
\ifx\versy\fully
For example, a semiautomatic group $(A,\circ;=)$ would be a structure where
the domain $A$ is regular, the group operation (with both inputs) is automatic
and for each fixed element $a \in A$ the set $\{b \in A: b = a\}$ is
regular --- note that group elements might have several representatives in
semiautomatic groups.
\fi
The present work will focus more on structures
like rings than groups, although the field of automatic and semiautomatic
structures has a strong group theoretic component.
The construction of these semiautomatic rings is similar to that of
Nies and Semukhin \cite{NS09} for a presentation of ${\mathbb Z}^2$ where
no $1$-dimensional subgroup is a regular subset.

The interested reader finds information about automatic structures
in the surveys of Khoussainov and Minnes \cite{KM10} and
Rubin \cite{Ru08}. Related but different links between automata theory
and geometry have been studied previously like, for example, the usage
of weighted automata and transducers to generate fractals \cite{CK97},
$\omega$-automata to represent geometric objects in the reals
\cite{BJW05,JSY07}
and the field of reals not being $\omega$-automatic
\ifx\versy\confy
\cite{ZGKP13}.
\fi
\ifx\versy\fully
\cite{ZGK12,ZGKP13}.
\fi
The last section of the present work applies the results and methods
of the current work to $\omega$-automatic structures.

The present work looks at semiautomatic rings which can be used
to represent selected points in the real plane. Addition and
subtraction and comparisons
\pagebreak[3]
as well as multiplication with constants
have to be automatic; however, the full multiplication is not automatic.
It depends on the structures which geometric objects and operations
with such object can be represented.

\begin{defn} \label{de:automatic}
The {\em convolution} of two words $v,w$ is a mapping from the
union of their domains to $(\Sigma \cup \{\#\}) \times (\Sigma \cup \{\#\})$
such that first one extends $v,w$ to $v',w'$, each having the domain
$dom(v) \cup dom(w)$, by assigning $\#$ whenever $v$ or $w$ are undefined
and then letting the convolution $u$ map every $h \in dom(v) \cup dom(w)$
to the new symbol $(v'(h),w'(h))$. Similarly one defines the convolutions
of three, four or more words.

A $h$-ary relation $R$ is {\em automatic}
\ifx\versy\fully
\cite{Bl99,BG00,Ho76,Ho83,KN95}
\fi
\ifx\versy\confy
\cite{BG00,Ho76,Ho83,KN95}
\fi
iff the set of all the convolutions of tuples 
$(x_1,\ldots,x_h) \in R$ is regular; a $h$-ary function $f$ is automatic
iff the set of all convolutions of $(x_1,\ldots,x_h,y)$ with
$f(x_1,\ldots,x_h) = y$ is regular. A $h$-ary relation $P$ is 
{\em semiautomatic} \cite{JKSTZ17}
iff for all indices $i \in \{1,\ldots,h\}$ and for all possible fixed values
$x_j$ with $j \neq i$ the resulting set $\{x_i: (x_1,\ldots,x_h) \in P\}$
is regular. A $h$-ary function $g$ is semiautomatic iff for all
indices $i \in \{1,\ldots,h\}$ and all possible values $x_j$ with $j \neq i$
the function $x_i \mapsto g(x_1,\ldots,x_h)$ is automatic.

A structure $(A,f_1,\ldots,f_k,R_1,\ldots,R_\ell;g_1,\ldots,g_i,
P_1,\ldots,P_j)$ 
is {\em semiautomatic} \cite{JKSTZ17}
iff (i) $A$ is a regular set of words where each word maps a finite subset
of $\mathbb Z$ to a fixed alphabet, (ii) each $f_h$ is
automatic, (iii) each $R_h$ is automatic, (iv) each $g_h$ is semiautomatic
and (v) each $P_h$ is semiautomatic. The semicolon separates the automatic
components of the structure from those which are only semiautomatic.
Structures without semiautomatic items are just called automatic.

An {\em automatic family}
\ifx\versy\fully
\cite{JOPS10}
\fi
$\{L_d: d \in E\}$ is a collection of sets such that
their index set $E$ and the set of all convolutions of $(d,x)$ with
$x \in L_d$ and $d \in E$ are regular.
\end{defn}

\begin{defn}
A {\em semiautomatic grid} or, in this paper, just grid,
is a semiautomatic ring $(A,+,=,<;\cdot)$ where the multiplication
is only semiautomatic and the addition and comparisons are automatic
such that $A$ forms a dense subring of the reals, that is, whenever $p,r$
are real numbers with $p<r$ then there is an $q \in A$ with
$p < q \wedge q < r$ and furthermore, all elements of $A$ represent
real numbers.
\end{defn}

\noindent
It makes sense to define density as a property of an ordered ring that
is embeddable into the reals, since the embedding is unique. A necessary
and sufficient criterion for the ring to be dense is that it has an
element strictly between $0$ and~$1$.
\ifx\versy\fully
If there is $q$ with $0<q<1$ and $p,r$ are any ring members with
$p<r$ then $p + q \cdot (r-p)$ is strictly between $p$ and $r$.
Furthermore, ordered rings which are subrings of the real numbers
do not have an endpoint, as adding $+1$ or $-1$ to any element
allows finding elements strictly above and strictly below the element.
\fi

\begin{exmp} \label{gridexample}
The ring $({\mathbb D}_b,+,=,<;\cdot)$ of the rational numbers in base $b$
with only finitely many nonzero digits is a grid. Here ${\mathbb D}_b =
\{n/b^m: n,m \in {\mathbb Z}\}$. Addition and comparison follow the
school algorithm as in the following example of Stephan \cite{St15}.
In ${\mathbb D}_{10}$, given three numbers $x,y,z$,
an automaton to check whether $x+y = z$ would process from the back
to the front and the states would be ``correct and carry to next digit (c)'',
``correct and no carry to next digit (n)'' and ``incorrect (i)''.
In the following three examples, $x$ stands on the top, $y$ in the
second and $z$ in the last row. The states of the automaton are for
starting from the end of the string to the beginning after having
processed the digits after them but not those before them. The filling
symbol $\#$ is identified with $0$. The decimal dot is not there
physically, it just indicates the position between digit
$a_0$ and digit $a_{-1}$.
The domain of each string is an interval from
a negative to a positive number plus an entry for the sign $-$ if needed.

\begin{verbatim}
   Correct Addition     Incorrect Addition  Incomplete Addition
    # 2 3 5 8. 2 2 5     3 3 3 3. 3 3 #      9 9 1 2 3. 4 5 6
    # 9 1 1 2. # # #     # # 2 2. 2 2 2      # # 9 8 7. 6 5 4
    1 1 4 7 0. 2 2 5     # 1 5 5. 5 5 2      0 0 1 1 1. 1 1 #
   n c n n c  n n n n   i i n n  n n n n    c c c c c  c c c n
\end{verbatim}

The difference
$x-y=z$ is checked by checking whether $x=y+z$ and then one can compare
the outcome of additions of possibly negative numbers by going to $-$
when the signs of the numbers require this. Furthermore, $x<y$ iff
$y-x$ is positive and $x=y$ if the two numbers are equal as strings.

For checking whether $x \cdot i/j = y$ for given rational constant $i/j$,
one just checks whether $i \cdot x = j \cdot y$ which, as $i,j$ are
constants, can be done by $i$ times adding $x$ to itself and $j$ times
adding $y$ to itself and then comparing the results. So $x \cdot 3/2 = y$
is equivalent to $x+x+x = y+y$ and the latter check is automatic.
Also the set of all $x \in {\mathbb D}_{10}$ so that $x$ is a multiple
of $3$ is regular, as it is first-order definable as
$
   \{x: \exists y \in {\mathbb D}_{10}\,[x = y+y+y]\}
$
and $1.2$ would be in this set and $1.01$ not. This works for all multiples
of fixed rational numbers in ${\mathbb D}_b$.
\end{exmp}

\section{Grids with Special Properties}

\noindent
Jain, Khoussainov, Stephan, Teng and Zou \cite{JKSTZ17} showed that
for every natural number $c$ which is not a square there is a grid
containing $\sqrt{c}$. Though these grids are dense subsets of the
real numbers, they do not have the property that one can divide by
any natural number, that is, for each $b \geq 2$ there is a ring element
$x$ such that $x/b$ is not in the ring. The reason is that most of the
rings considered by Jain, Khoussainov, Stephan, Teng and Zou are of
the form ${\mathbb Z} \oplus \sqrt{c} \cdot {\mathbb Z}$.
The following result will produce grids for which one can always divide
by some number $b \geq 2$, if this number is composite, it might allow
division by finitely many primes. Note that the number of primes cannot
be infinite by a result of Tsankov \cite{Ts11}.

\begin{thm} \label{th:main}
Assume that $b \in \{2,3,4,\ldots\}$ and $c$ is some root of an integer
and let $u>1$ be a real number chosen such that
the following four polynomials $p_1,p_2,p_3,p_4$ in a variable $x$
and constants $\ell,\hat c$ exist, where all polynomials have only
finitely many nonzero coefficients and all coefficients are integers:
\begin{enumerate}
\item $p_1(u) = \sum_{k \in \mathbb Z} b_k u^k = 1/b$;
\item $p_2(u) = \sum_{k \in \mathbb Z} c_k u^k = c$;
\item $p_3(u) = \sum_{k=0,-1,-2,\ldots,-h+1} d_k u^k = 0$ with $d_0=1$;
\item $p_4(u) = \sum_{k \in \mathbb Z} e_k u^k = 0$ with
      $e_\ell > \sum_{k \neq \ell} |e_k|$ and $| e_k |$ being
      the absolute value of $e_k$.
\end{enumerate}
Furthermore, the choice of the above has to be such that $\hat c > 3|e_\ell|$
and one can run for every polynomial
$p = \sum_{k=-m,\ldots,n} a_k x^k$ with every $a_k$ being an integer
satisfying $|a_k| \leq 3 |e_\ell|$ the following algorithm $C$ satisfying
the below termination condition:
\begin{quote}
  Let $k = n+h$. \\
  While $k > -m$ and $|a_{k'}| \leq \hat c$ for $k'=k,k-1,\ldots,k-h+1$ \\
  Do Begin $p = p-a_k \cdot p_3(u) \cdot u^k$
  and update the coefficients of the polynomial $p$ accordingly; Let $k=k-1$ End.
\end{quote}
The termination condition on $C$ is that whenever the algorithm terminates
at some $k > -m$ with some $|a_{k'}| > \hat c$ then
\begin{quote}
$|\sum_{k'=k,k-1,\ldots,k-h+1} u^{k'} a_{k'}| >
  u^{k-h}/(1-u^{-1}) \cdot 3|e_{\ell}|$.
\end{quote}
If all these assumptions are satisfied then one can use the representation
$$
S = \{\sum_{k=-m,\ldots,n} a_k u^k: m,n \in {\mathbb N}, \ 
a_k \in {\mathbb Z} \mbox{ and } |a_k| < |e_\ell|\}
$$
to represent every member of ${\mathbb D}_b[c]$ and
the ring $(S,+,<,=;\cdot)$ has automatic addition and comparisons and
semiautomatic multiplication. Furthermore, as $1/b$ is in the ring, it is
a dense subset of the reals, thus the ring forms a semiautomatic grid.
\end{thm}

\begin{proof}
When not giving $-m,n$ explicitly in the sum, sums like
$\sum_{k \in {\mathbb Z}} a_k u^k$ use the assumption
that almost all $a_k$ are $0$.
For ease of notation, let $S'$ be the set
$$
S' = \{\sum_{k \in {\mathbb Z}} a_k x^k:
       \mbox{ almost all $a_k$ are $0$ and all }
       a_k \in {\mathbb Z}\} \vspace{-0.1cm}
$$
so that $S \subseteq S'$.
On members $p,q \in S'$, one defines that $p \leq q$ iff $p(u) \leq q(u)$
when the polynomial is evaluated at the real number $u$. Furthermore,
$p = q$ iff $p \leq q$ and $q \leq p$. Addition and subtraction in $S'$
is defined using componentwise addition of coefficients.

Now one shows that for every $p \in S'$ there is a $q \in S$ with
$p = q$. For this one lets initially $h=0$ and $q_h = p$ and
whenever there is a coefficient $a_k$ of $q_h$ with $|a_k| \geq |e_\ell|$
then one either lets $q_{h+1} = q_h - x^{k-\ell} \cdot p_4$ (in the case
that $a_k > 0$) or lets $q_{h+1} = q_h + x^{k-\ell} \cdot p_4$ (in the
case that $a_k < 0$). Now let $||q_h||$ be the sum of the absolute values
of the coefficients; note that
$$
   ||q_{h+1}|| \leq ||q_h||-e_\ell+\sum_{k \neq \ell} |e_k| < ||q_h||
$$
and as there is no infinite strictly decreasing sequence of integers,
there is a $h$ where $q_h$ is defined but $q_{h+1}$ not,
as this update can no longer be made. Thus all coefficients of $q_h$
are between $-|e_\ell|$ and $+|e_\ell|$ and furthermore, as each
polynomial $p_4(u) \cdot u^{k-\ell}$ added or subtracted has the value $0$,
$q_h = p$. Now let $q = q_h$ and note that $q$ is a member of $S$ with
the same value at $u$ as $p$, so $p(u) = q(u)$.

Now let $p,q,r$ be members of $S$. In order to see what the sign
of $p+q-r$ is, that is, whether $p(u)+q(u)<r(u)$, $p(u)+q(u)=r(u)$
or $p(u)+q(u)>r(u)$, one adds the coefficients pointwise and
to check the expression $p+q-r$ at $u$, one then runs the algorithm $C$.
If $C$ terminates with some $|a_{k'}| > \hat c$, then the sign of
the current value of
\ifx\versy\fully
$$\sum_{k'=k,k-1,\ldots,k-h+1} u^{k'} a_{k'}$$
\fi
\ifx\versy\confy
$$\tilde a = \sum_{k'=k,k-1,\ldots,k-h+1} u^{k'} a_{k'}$$
\fi
gives the sign of $p+q-r$, as the not yet processed tail-sum
of $p+q-r$ is bounded by $u^{k-h}/(1-u^{-1}) \cdot 3|e_{\ell}|$.
In the case that $C$ terminates with all $|a_{k'}| \leq \hat c$ and
$k=-m$, then only the coefficients at $k'=k,k-1,\ldots,k-h+1$ are not
zero and again the sign of
\ifx\versy\fully
$$\sum_{k'=k,k-1,\ldots,k-h+1} u^{k'} a_{k'}$$
\fi
\ifx\versy\confy
$\tilde a$
\fi
is the sign of the original polynomial $p+q-r$.

Note that the algorithm $C$ can be carried out by a finite automaton,
as it only needs to memorise the current values of
$(a_k,a_{k-1},$ $\ldots,a_{k-h+1})$ which are $(0,0,\ldots,0)$ at the
start and which are updated in each step by reading $a_{k-h}$ for
$k=n+h,n+h-1,\ldots,-m$; the update is just subtracting
$a_{k'} = a_{k'}-a_k \cdot d_{k'-k}$ for $k'=k,k-1,\ldots,k-h+1$
and then updating $k=k-1$ which basically requires to read $a_{k-h}$ into the
window and shift the window by one character; note that the first
member, which goes out of the window, is $0$. Furthermore, during
the whole runtime of the algorithm, all values in the window have
at most the values $(1+\max\{|d_{k''}|: 0 \geq k'' \geq -h+1\}) \cdot
\hat c$ and thus there are only finitely many choices for
$(a_k,a_{k-1},\ldots,a_{k-h+1})$, and thus the determination
of the sign of $\sum_{k'=k,k-1,\ldots,k-h+1} u^{k'} a_{k'}$ 
can be done by looking up a finite table. Early
termination of the finite automaton can be handled by not changing
the state on reading new symbols, once it has gone to a state with some
$|a_{k'}| > \hat c$. Thus comparisons and addition are automatic;
note that for automatic functions, the automaton checks whether the
tuple $(inputs,output)$ is correct, it does not compute $output$
from $inputs$.

For the multiplication with constants, note that multiplication with
$u$ or $u^{-1}$ is just shifting the coefficients in the representation
by one position; multiplication with $-1$ can be carried out componentwise
on all coefficients; multiplication with integers is repeated addition
with itself. This also then applies to polynomials put together from
these ground operations, so $p \cdot (u^2-2+u^{-1})$ can be put together
as the sum of $p \cdot u \cdot u$, $-p$, $-p$, $p \cdot u^{-1}$.
All four terms of the sum can be computed by concatenated automatic
functions, thus there is an automatic function which also computes
the sum of these terms from a single input~$p$.~$\Box$
\end{proof}

\begin{exmp} \label{2R2+1/2}
There is a semiautomatic grid containing $\sqrt{2}$ and $1/2$.
\end{exmp}

\begin{proof}
For $c = \sqrt{2}$ and $b=2$, one chooses
\begin{enumerate}
\item $u^{-1} = 1-c/2$ (note that $u = 1/(1-\sqrt{1/2}) = 2/(2-\sqrt{2})>1$),
\item $p_1(u) = 2u^{-1}-u^{-2} = 1/2$,
\item $p_2(u) = 2-2u^{-1} = c$,
\item $p_3(u) = 1-4u^{-1}+2u^{-2} = 0$,
\item $p_4(u) = -u+4-2u^{-1} = 0$ with $\ell = 0$ and $e_{\ell}=4$,
\item $\hat c = 100$ (or any larger value).
\end{enumerate}
While all operations above come from straight-forward manipulations
of the choice of $u^{-1}$, one has to show the termination condition
of the algorithm.

For this one uses that $u \geq 3.41$ and
$1/(1-1/u) = \sum_{k \leq 0} u^k = \sqrt{2} \leq 1.4143$.
Assume that the algorithm satisfies before doing
the step for $k$ that all $|a_{k'}| \leq \hat c$ and does not
satisfy this after updating $a_k,a_{k-1},a_{k-2}$ respectively to $0$,
$a' = a_{k-1}+4a_k$ and $a'' = a_{k-2}-2a_k$; in the following,
$a_k,a_{k-1},a_{k-2}$ refer to the values before the update.
Without loss of generality assume
that $a_k > 0$, the case $a_k < 0$ is symmetric, the case $a_k=0$
does not make the coefficients go beyond $\hat c$.
If $a'' < -\hat c$ --- it can only go out of the range to the negative side
--- then $2a_k \geq \hat c-3 (e_\ell-1)$ and $p(u)$ is at least
$a_k \cdot (4-2/u) \cdot u^{k-1} - \hat c \cdot u^{k-1}/(1-1/u)
  \geq ((2-1/u) \cdot  (\hat c -9) - 1.4143 \hat c) u^{k-1}
  \geq (1.7 \cdot (\hat c \cdot 0.9) - 1.4143 \hat c) u^{k-1}
  \geq 0.1 \cdot \hat c \cdot u^{k-1}> 0$.
If $a'' \geq -\hat c$ and $a' > \hat c$ then
$p(u) \geq (\hat c \cdot u - 1/(1-1/u) \hat c) \cdot u^{k-2} \geq
\hat c \cdot u^{k-2} > 0$. So in both cases, one can conclude that
$p(u)$ is positive. Similarly, when $a_k < 0$ and the bound $\hat c$
becomes violated in the updating process then \mbox{$p(u) < 0$.}~$\Box$
\end{proof}

\begin{exmp} \label{2R3+1/2}
There is a grid which contains $\sqrt{3}$ and $1/2$ or, more generally,
any $c$ of the form $c = \sqrt{b^2-1}$ and $1/b$
for some fixed integer $b \geq 2$.
\end{exmp}

\begin{proof}
One chooses
\begin{enumerate}
\item $u^{-1} = 1-c/b$ (note that $u = b/(b-c) > 1$),
\item $p_1(u) = 2bu^{-1}-bu^{-2} = 1/b$,
\item $p_2(u) = b-bu^{-1} = c$,
\item $p_3(u) = 1-2b^2u^{-1}+b^2u^{-2} = 0$,
\item $p_4(u) = -u+2b^2-b^2u^{-1} = 0$ with $\ell = 0$ and $e_{\ell}=2 b^2$,
\item $\hat c = 1000 \cdot b^5$ (or any larger value).
\end{enumerate}
While all operations above come from straight-forward manipulations
of the equations, the termination condition of the algorithm needs
some additional work. Note that $u > b$, as $b-c < 1$.
Indeed, by $u \geq 1$ and $p_3(u)=0$ and $b \geq 2$, one has
$1-b^2u^{-1} \geq 0$ and $u \geq b^2$ and $\sum_{k \leq 0} u^k \leq 2$.
For the algorithm, one now notes that if after an update at $k$
where, without loss of generality, $a_k > 0$, it happens that
either (a) $a'' = a_{k-2}-b^2 a_k < -\hat c$ or (b) $a'' \geq -\hat c$ and
$a' = a_{k-1}+2b^2 a_k > \hat c$ then the following holds:
In the case (a), $a_k \geq 1000 b^3 - 6 b^2$ and the value of the sum
is at least
\begin{eqnarray*}
   (a_k \cdot (2b^2 u-b^2)-\hat c \cdot u - 12 b^4) \cdot u^{k-2} & > & \\
   ((2000 b^5 u -6b^4) - 1000 b^5\cdot u - 12 b^4) \cdot u^{k-2} & > & \\
   (1000 b^5 - 18 b^4) \cdot u^{k-1} & > & 0
\end{eqnarray*}
where the tail sum $12 b^4$ estimates that all digits $a_{k'}$ with
$k' \leq k-2$ are at least $-6b^2$ in the expression and the
$a_{k-1}$ is at least $-\hat c$ by assumption. In case (b),
one just uses that the first coefficient in the sum is greater than
$\hat c$ while all other coefficients are of absolute value
below $\hat c$, in particular as $\hat c \geq 6b^2$, so that, since $u \geq 2$,
$$
   \hat c \cdot u^{k-1} > \sum_{k' < k-1} \hat c u^{k'} \mbox{ and } 
   \hat c \cdot u^{k-1} > 2 \cdot \hat c \cdot u^{k-2}.
$$
Thus the algorithm terminates as required.~$\Box$
\end{proof}

\section{Applications to Geometry}

\noindent
One can use the grid to represent the coordinates of geometric
objects. For this, one uses in the field of automatic structures
the concept of convolution which uses the overlay of constantly
many words into one word. One introduces a new symbol, $\#$, which
is there to pad words onto the same length. Now, for example,
if in the grid of decimal numbers, one wants to describe a point
of coordinates $(1.112,22.2895)$, this would be done with the
convolution
$(\#,2) (1,2) . (1,2) (1,8) (2,9) (\#,5)$ where these six characters
are the overlay of two characters and the dot is virtual and
only marking the position where the numbers have to be aligned,
that means, the position between the symbols at location $0$ and
location $-1$. Instead of combining two numbers, one can also combine
five numbers or any other arity.

An automatic family is a family of sets $L_e$ with the indices
$e$ from some regular set $D$ such that the set $\{conv(e,x): x \in L_e\}$
is regular. Given a grid $G$, the set of all lines parallel to the
$x$-axis in $G \times G$ is an automatic family:
Now $D = G$ and $L_y = \{conv(x,y): x \in G\}$. The next example
shows that one cannot have an automatic family of all lines.

\begin{exmp}
The set of all lines (with arbitrary slope) is not an automatic
family, independent of the definition of the semiautomatic grid.
Given a grid $G$ and assuming that $\{L_e: e \in D\}$ is the
automatic family of all lines, one can first-order define the
multiplication using this automatic family:
\begin{quote}
  $x \cdot y = z$ if either at least one of $x,y$ and also $z$ are $0$
  or $x=1$ and $y=z$ or $y=1$ and $x=z$
  or all are nonzero and neither $x$ nor $y$ is $1$
  and there exists an $e \in D$ such that $conv(0,0),conv(1,y),conv(x,z)$
  are all three in $L_e$.
\end{quote}
As the grid $G$ has to be dense and is a ring with automatic addition
and comparison and as $G \subset \mathbb R$, the ring $G$ is an integral
domain and furthermore, $G$ has an automatic multiplication by the above
first-order definition. Khoussainov, Nies, Rubin and Stephan \cite{KNRS04}
showed that no integral domain is automatic, hence the collection of all
lines cannot be an automatic family, independent of the choice of the
grid.
\end{exmp}

\noindent
Similarly one can consider the family of all triangles.

\begin{thm}
Independently of the choice of the semiautomatic grid $G$,
the family of all triangles in the plane is not an automatic family.
However, every triangle with corner points in $G \times G$ is
a regular set.
\end{thm}

\begin{proof}
For the first result, assume that $\{L_e: e \in D\}$ is a family
of all triangles -- when viewed as closed subsets of $G \times G$ --
which are represented in the grid and that this family contains at least
all triangles with corner points in $G \times G$. Now one can define for
$x,y,z > 0$ the multiplication-relation $z = x \cdot y$ using this family
as follows:
\begin{quote}
   $z = x \cdot y$ $\Leftrightarrow$ some $e\in D$ satisfies
   the following conditions: \\
   $\forall v,w \in G \mbox{ with } v \leq 0\,[(v,w) \in L_e
                     \Leftrightarrow (v,w) = (0,0)]$, \\
   $\forall w \in G\,[(1,w) \in L_e
                     \Leftrightarrow 0 \leq w \wedge w \leq y]$, \\
   $\forall w \in G\,[(x,w) \in L_e
                     \Leftrightarrow 0 \leq w \wedge w \leq z]$.
\end{quote}
This definition can be extended to a definition for the multiplication
on full $G$ with a straightforward case-distinction.
Again this cannot happen as then the grid would form an infinite
automatic integral domain which does not exist.

However, given a triangle with corner points $(x,y),(x',y'),(x'',y'')$,
note that one can find that linear functions from $G \times G$
into $G$ which are nonnegative iff the input point is on the right side
of the line through $(x,y)$ and $(x',y')$. So one would require that the
function
$$f(v,w) = (w-y) \cdot (x'-x) - (v-x) \cdot (y'-y) $$
is either always nonpositive or always nonnegative, depending on which
side of the line the triangle lies; by multiplying $f$ with $-1$, one can
enforce nonnegativeness. Note here that $x'-x$ and $y'-y$ are constants
and multiplying with constants is automatic, as the ring has a semiautomatic
multiplication. Thus a point is in the triangle or on its border iff
all three automatic functions associated with the three border-lines
of the triangle do not have negative values. This allows to show
that every triangle with corner points in $G \times G$ is regular.~$\Box$
\end{proof}

\medskip
\noindent
Note that this also implies that polygons with all corner points
being in $G \times G$ are regular subsets of the plane $G \times G$.

\begin{prop}
Moving a polygon by a distance $(v,w)$ can be done in any grid,
as it only requires adding $(v,w)$ to the coordinates of each points.
However, rotating by $30^\circ$ or $45^\circ$ is possible in some
but not all grids.
\end{prop}

\begin{proof}
Note that the formula for rotating around $30^\circ$, one needs to
map each point $(x,y)$ by the mapping
$(x,y) \mapsto (\cos(30^\circ) x - \sin(30^\circ)y,
\sin(30^\circ) x + \cos(30^\circ) y)$ and similarly for $45^\circ$
and $60^\circ$. For $30^\circ$, as $\sin(30^\circ) = 1/2$ and
$\cos(30^\circ) = \sqrt{3}/2$, one needs a grid which allows to divide
by $2$ and multiply by $\sqrt{3}$, an example is given by the grid
of Example~\ref{2R3+1/2}. For rotating by $60^\circ$, as it is twice
doing a rotation by $30^\circ$, the same requirements on the grid need
to be there. For rotating by $45^\circ$, the grid from Example~\ref{2R2+1/2}
can be used. However, these operations cannot be done with grids which
do not have the corresponding roots and also do not have the possibility
to divide by $2$. In particular, the grids ${\mathbb D}_b$ do not allow
to multiply by roots and those grids of the form
${\mathbb Z} \cup \sqrt{c} \cdot
{\mathbb Z}$ in prior work \cite{JKSTZ17} do not allow to divide by $2$.
Furthermore, the authors are not aware of any grid which has both,
$\sqrt{2}$ and $\sqrt{3}$ and thus allow to rotate by
both, $30^\circ$ and $45^\circ$.~$\Box$
\end{proof}

\begin{rem}
A grid allows to represent all equilateral triangles with side-length in $G$
iff $\sqrt{3}$ and $1/2$ are both in $G$.
\ifx\versy\fully
This is in particular true for grids which allow to rotate by $60^\circ$
and it is false for all grids of the type ${\mathbb D}_b$.

For a proof, assume that in a plane $G \times G$, an equilateral triangle is
represented with all three corner points in $G \times G$.
So let $(x,y),(x',y')$ be two corner points in $G \times G$. Now the
third corner point $(x'',y'')$ has either the coordinates
$$((x+x')/2+(y-y') \cdot \sqrt{3}/2,(y+y')/2+(x'-x)\cdot \sqrt{3}/2)$$
or the coordinates
$$((x+x')/2-(y-y') \cdot \sqrt{3}/2,(y+y')/2-(x'-x)\cdot \sqrt{3}/2)$$
and in either case, a scaling of either $x-x'$ or $y-y'$ by $\sqrt{3}/2$
has to be in $G$ and at least one of these is nonzero.

This is impossible in the case of $G = {\mathbb D}_b$, as all members
of $G$ are rational.

However, it can always be done in the grid from Example~\ref{2R3+1/2}
with $c = \sqrt{3}$ and $b=2$, as in that grid there are automatic
functions which divide the input by $2$ and which multiply the input
with $\sqrt{3}$, respectively. In this grid, for any base-line given
by two distinct points in $G \times G$, one can find the third point in order
to obtain the equilateral triangle with the three corner points.
\fi
\end{rem}

\begin{rem} \label{rem:standardrep}
Note that one represents a word $a_5 a_4 \ldots a_1
a_0.a_{-1} a_{-2}\ldots a_{-7}$ also by starting with $a_0$ and then
putting alternatingly the digits of even and odd indices giving
\ifx\versy\fully
$$a_0 a_1 a_{-1} a_2 a_{-2} \ldots a_5 a_{-5} 0 a_{-6} 0 a_{-7}$$
\else
$a_0 a_1 a_{-1} a_2 a_{-2} \ldots a_5 a_{-5} 0 a_{-6} 0 a_{-7}$
\fi
and one can show that in this representation,
the same semiautomaticity properties are valid as in the previously considered
representation. However, one gets one additional relation: One can recognise
whether two digits $a_{-m}$ and $a_n$ satisfy that $n=m+c$ for a given
integer constant $c$. This is used in the following example.
\end{rem}

\begin{exmp}
The family $\{E_d: d \in D\}$ of all axis-parallel rectangles is an
automatic family in all grids. Furthermore, let $d \equiv d'$ denote
that $E_d$ and $E_{d'}$ have the same area. In no grid, this relation
$\equiv$ is automatic, as otherwise one could reconstruct the multiplication.

For $p$ being a prime power,
in the grid $({\mathbb D}_p,+,=,<;\cdot)$ from Example~\ref{gridexample}, 
the relation $\equiv$ is semiautomatic
using the representation given in Remark~\ref{rem:standardrep}.
To see this, for a given area $\ell \cdot p^k$,
(i) one can disjunct over
all factorisations $\ell_1 \cdot \ell_2$ of $\ell$ which are pairs
of natural numbers not divisible by $p$, then (ii) check whether the length
of the sides of a given rectangle are of the form $\ell_1 \cdot p^i$
and $\ell_2 \cdot p^j$ with $i+j = k$, where $i,j$ are the positions of
the last
nonzero digits in the $p$-adic representation of the lengths and
(iii) check, by Remark~\ref{rem:standardrep}, whether $i+j$ are
the given constant $k$. Note that 
representations using prime powers can be translated
into representations based on the corresponding primes.

However, for grids such as $({\mathbb D}_6,+,=,<;\cdot)$, where $b$ is a
composite number other than a prime power, this method does not work.
This is mainly because one needs to use base $6$
for the comparison $<$
and then a finite automaton cannot see whether the two sides are, for
example, of lengths $2^k$ and $2^{-k}$ when recognising squares of area $1$.
Knowing that this method does not work, however, does not say that 
no other method works. It is an open problem whether one can
find a semiautomatic grid which allows to divide by $6$ and to represent
axis-parallel rectangles in a way such that 
checking whether two rectangles have same area is semiautomatic.
The same applies to the grids
of Examples~\ref{2R2+1/2} and~\ref{2R3+1/2}.
\end{exmp}

\section{Cube Roots}

\noindent
Jain, Khoussainov, Stephan, Teng and Zou \cite{JKSTZ17} did not find
any example of other roots than square roots to be represented in
a semiautomatic ordered ring. The following example represents a cube root.

\begin{exmp} \label{3R7}
There is a grid which contains $\sqrt[3]{7}$. Furthermore, there is
a grid which contains $\sqrt[3]{65}$ and $1/2$.

For the first, as one does not want to represent a proper rational,
$p_1$ is not needed. For this, one chooses
\begin{enumerate}
\item $u^{-1} = 2-\sqrt[3]{7}$,
\item $p_2(u) = 2-u^{-1} = \sqrt[3]{7}$,
\item $p_3(u) = 1-12u^{-1}+6u^{-2}-u^{-3} = 0$ and $p_4(u) = -p_3(u)$
      with $\ell = -1$,
\item Instead of a flat $\hat c$, one uses a bit different bound for the
      algorithm, namely $|a_k| \leq 16\hat c$, $|a_{k-1}| \leq 4 \hat c$
      and $|a_{k-2}| \leq \hat c$ for $\hat c = 360$.
\end{enumerate}
The equations for $p_3,p_4$ follow from $p_3(u) = (p_2(u))^3-7$.
Note that $11 < u < 12$ and $u^{-1}+u^{-2}+\ldots \leq 1/10$.
Furthermore, the coefficients in the normal form are between $-12$ and $+12$
and, when three numbers are added coefficientwise, between $-36$ and $+36$.
Let $p = \sum_k a_k \cdot u^k$ be such a sum of three numbers whose
sign has to be determined; all the coefficients have absolute values up
to $36$.

The main thing is that the algorithm can detect the sign of the number
whenever the first three coefficients overshoot for the first time.
Note that they start with $(0,0,0)$ and so one runs the updates
$a_{k-k'} = a_{k-k'}-a_k \cdot d_{-k'}$ simultaneously for $k'=1,2,3$
and then sets $a_k = 0$ and $k=k-1$
Here $d_{\cdot}$ are coefficients of $p_3$.
Assume that the update would make
the coefficients to overshoot for the first time and let $k,a_k,a_{k-1},
a_{k-2},a_{k-3}$ and $p = \sum_{k'} a_{k'}u^{k'}$
denote the values just before the update.

Without loss of generality, assume that $a_k > 0$ and it will be shown
that this implies that the polynomial $p$ would be positive.
Note that before the update,
for all $k'<k$, $|a_{k'}| \leq 4 \hat c$ and
$|\sum_{k'<k} a_{k'} u^{k'}| \leq 0.4 \cdot \hat c \cdot u^k$.

If $a_{k-3}$ grows above $\hat c$ at the update
then $a_k \geq \hat c \cdot 9/10$
and $p > (0.9 \cdot \hat c - \sum_{k' \leq -1} a_{k+k'} \cdot u^{k'}) \cdot u^k
       \geq (0.9 - 0.4) \cdot \hat c \cdot u^k > 0$.

If $a_{k-2}$ grows below $-4\hat c$ at the update
but $a_{k-3}$ stays inside the bound
then $a_k \geq \hat c \cdot 3 \cdot 1/6$
and $p > (0.5-0.4) \cdot \hat c \cdot u^k > 0$.

If $a_{k-1}$ grows beyond $16 \hat c$ at the update
but $a_{k-2}$ and $a_{k-3}$ stay inside their bounds
then $a_k \geq \hat c \cdot 12 \cdot 1/12$
and $p > (1-0.4) \cdot \hat c \cdot u^k > 0$.

\ifx\versy\confy
Thus there is a semiautomatic grid containing $\sqrt[3]{7}$.
One similarly proves the second item that there is
a semiautomatic grid containing $\sqrt[3]{65}$ and $1/2$.
\fi
\ifx\versy\fully
So the above test can detect the sign of $p$ by just looking at the sign
of $a_k$ in the event that the next step overshoots; if, through the
runtime, it never overshoots, then one can deduce the sign of $p$ from a
table look-up of the final values of the coefficients in the tracked
window. Thus there is a semiautomatic grid containing $\sqrt[3]{7}$.

Furthermore, for the grid which contains $\sqrt[3]{65}$ and $1/2$,
one chooses
\begin{enumerate}
\item $u^{-1} = (\sqrt[3]{65}-4)/2$, note that $96.49 \leq u \leq 96.50$,
\item $p_1(u) = 4((u^{-1}+2)^3-8) = 48u^{-1}+24u^{-2}+4u^{-3} = 1/2$,
\item $p_2(u) = 2u^{-1}+4 = \sqrt[3]{65}$,
\item $p_3(u) = 1-96u^{-1}-48u^{-2}-8u^{-3} = 0$ and $p_4(u) = -p_3(u)$
      with $\ell = -1$,
\item Instead of a flat $\hat c$, one uses a bit different bound for the
      algorithm, namely $|a_k| \leq 19\hat c$, $|a_{k-1}| \leq 7 \hat c$
      and $|a_{k-2}| \leq \hat c$ for $\hat c$ very large, for example
      $\hat c = 285000$.
\end{enumerate}
Note that $96.49 < u < 96.50$ and $u^{-1}+u^{-2}+\ldots \leq 1/95$.
Furthermore, the coefficients in the normal form are between $-95$ and $+95$
(both inclusive)
and, when three numbers are added coefficientwise, between $-285$ and $+285$
(both inclusive).
Let $p = \sum_k a_k \cdot u^k$ be such a sum of three numbers whose
sign has to be determined; all the coefficients have absolute values up
to $285$.

Again the algorithm can detect the sign of the number
whenever the first three coefficients overshoot for the first time.
Note that they start with $(0,0,0)$ and so one runs the updates
$a_{k-k'} = a_{k-k'}-a_k \cdot d_{-k'}$ simultaneously for $k'=1,2,3$
and then sets $a_k = 0$ and $k=k-1$. 
Here $d_{\cdot}$ are coefficients of $p_3$.
Assume that the update would make
the coefficients to overshoot for the first time and let $k,a_k,a_{k-1},
a_{k-2},a_{k-3}$ and $p = \sum_{k'} a_{k'}u^{k'}$
denote the values just before the update.

Without loss of generality, assume that $a_k > 0$ and it will be shown
that this implies that the polynomial $p$ would be positive.
Note that before the update,
for all $k'<k$, $|a_{k'}| \leq 7 \hat c$ and
$|\sum_{k'<k} a_{k'} u^{k'}| \leq 7/95 \cdot \hat c \cdot u^k \leq
0.074 \cdot \hat c \cdot u^k$.

If $a_{k-3}$ grows above $\hat c$ at the update
then $8a_k \geq \hat c \cdot 999/1000$ and $a_k \geq 0.1 \hat c$
and $p > (0.1 \cdot \hat c - \sum_{k' \leq -1} a_{k+k'} \cdot u^{k'}) \cdot u^k
       \geq (0.1 - 0.074) \cdot \hat c \cdot u^k > 0$.

If $a_{k-2}$ grows above $7\hat c$ at the update
then $48a_k \geq (7 \hat c - \hat c)$ and $a_k \geq \hat c/8 \geq 0.1 \hat c$
and $p > (0.1 - 0.074) \cdot \hat c \cdot u^k > 0$.

If $a_{k-1}$ grows beyond $19 \hat c$ at the update
then $96a_k \geq \hat c \cdot (19-7)$ and $a_k \geq \hat c/8 \geq 0.1 \hat c$
and $p > (0.1 - 0.074) \cdot \hat c \cdot u^k > 0$.

So the test works well and similarly the test also gives that $p<0$
in the case that there is an overshooting with negative $a_k$.
Thus there is a semiautomatic grid containing $\sqrt[3]{65}$ and $1/2$.
The method could be used to generate other examples of grids with
cube roots and fractions.
\fi
\end{exmp}

\section{Representing All Reals} \label{se:omega}

\noindent
\ifx\versy\fully
The authors assume for this section that the readers are
familiar with the theory of $\omega$-automata. Here $\omega$-words
are mappings $k \mapsto a_k$ from $\mathbb Z$ to the set of digits
such that, for some $k$, all $a_h$ with $h>k$ are $0$. An $\omega$-automaton
starts at an arbitrary $k$ such that all $a_h$ with $h \geq k$ are $0$
in a start state and then reads $a_k a_{k-1} a_{k-2} \ldots$ and
updates, on each $a_h$ read, its state in the same way as a finite automaton.
If the $\omega$-automaton goes infinitely often through an accepting state,
then it accepts the $\omega$-word, else it rejects the $\omega$-word.
In the case of determining the sign, there is a positive acceptance and
a negative acceptance; rejecting means that the $\omega$-word represents $0$.
Note that for an $\omega$-word to be accepted, the $\omega$-automaton needs
just to have some accepting run and not all runs need to be accepting.

\fi
Now a comment on the $\omega$-automatic approach \cite{BJW05,CK97,JSY07}.
The reals with addition and multiplication and infinite integral
domains in general are not $\omega$-automatic
\ifx\versy\fully
\cite{ZGK12,ZGKP13}.
\fi
\ifx\versy\confy
\cite{ZGKP13}.
\fi
It is also clear that $({\mathbb R},+;\cdot)$ is not
$\omega$-semiautomatic, as one otherwise would need uncountably many
different $\omega$-automata for recognising the uncountably many
functions $x \mapsto r \cdot x$ for constants $r \in \mathbb R$.
So the best what one can expect is that $({\mathbb R},+,<,=)$ is
$\omega$-automatic and that there are countably many functions
$f_r: x \mapsto r \cdot x$ which are $\omega$-automatic as well. These
functions certainly include, independent of the ring representation, all
$f_r$ with $r \in \mathbb Q$, as one only verifies the relation
$x \cdot r = y$ and for $r = i/j$ this is equivalent with verifying
$x \cdot i = y \cdot j$ with $i,j$ are integers and multiplication
with integers can be realised by repeated self-addition. The following
result shows that one can carry over ideas of Theorem~\ref{th:main}
to $\omega$-automatic structures in order to get that multiplication
with all constants from ${\mathbb Q}[\sqrt{b}]$ are $\omega$-automatic for
all natural numbers $b \geq 2$. As there is no need to implement the
division by $2$ explicitly, as that comes for free, one can use
the previously known representations \cite[Theorem~26]{JKSTZ17}.

\begin{thm} \label{th:omega}
There is an $\omega$-automatic representation of the reals where
addition, subtraction and comparisons are $\omega$-automatic and
furthermore the multiplication with any constant from
${\mathbb Q}[\sqrt{b}]$ is also an $\omega$-automatic unary function.
\end{thm}

\ifx\versy\fully
\begin{proof}
One uses the representation of \cite[Theorem 26]{JKSTZ17}. For
$b \geq 2$ which is not a square, one can use that there are
natural numbers $d,e>3$ with $d^2 = be^2+1$. Now one chooses
$u = d+e \sqrt{b}$ and as shown by Jain, Khoussainov, Stephan, Teng
and Zou \cite[Theorem 26]{JKSTZ17}, $u+u^{-1} = 2d$. Given any
$\omega$-word $\sum_{k \in \mathbb Z} a_k u^k$ where all sufficiently
large $k$ satisfy $a_k = 0$, one can always choose the largest coefficient
with $|a_k| > 2d$ and let $s$ be the sign of $a_k$ and
update $a_{k+1} = a_{k+1}+s, a_k = a_k - 2ds, a_{k-1}=a_{k-1}+s$;
either this algorithm terminates or runs forever.

In the case that all coefficients are initially bounded by
a constant $\hat c$, the algorithm produces in the limit an $\omega$-word
of the same value $\sum_k a_k u^k$ such that all $k$ satisfy $|a_k| \leq 2d$:
For the verification, one takes a value $\tilde u > 1$ such that
$2d-\tilde u - 1/\tilde u > 0$. Starting with an $\omega$-word
$p_0(u)$, let $p_n(u)$ be the $\omega$-word after the $n$-th iteration
and let $q_n(\tilde u)$ be $\sum_{k \in \mathbb Z} |a_k| \tilde u^k$
where the coefficients are taken from
$p_n(u) = \sum_{k \in \mathbb Z} a_k u^k$.
Due to the absolute bounds on the coefficients and due to
$a_k = 0$ for all sufficiently large $k$, both $p_n(u)$ and
$q_n(\tilde u)$ converge absolutely.
Note that the value $p_n(u) = \sum_k a_k u^k$ is for all $n$ the
same, as every $u^{k+1} - 2du^k + u^{k-1}$ has the value $0$
and $p_{n+1}(u) = p_n(u) + s(u^{k+1} - 2du^k + u^{k-1})$
for some $k$ and $s$. Furthermore,
$$
   q_{n+1}(\tilde u) \leq q_n(\tilde u) - (2d-\tilde u-1/\tilde u) \cdot 
    {\tilde u}^k
$$
for the $k$ from the update of $p_n$. Furthermore, $q_n(\tilde u) \geq 0$
for all $n$. For that reason, the sequence of the $q_n(\tilde u)$ converges
from above to some number. Thus, for each $k$ there can only be
finitely many updates where some digit at or above $k$ changes.
Thus the pointwise limit $p_{\infty}$ of all $p_n$ exists and
this limit satisfies that all coefficients have at most the absolute
value $2d$. Furthermore, all $p_n(u)$ have the same value. Furthermore,
if one looks at the difference $\omega$-words $p_n-p_\infty$ then its
coefficients converge at $\tilde u$ absolute and therefore converge
also absolute when $u$ is taken, as $u^{-1} < \tilde u^{-1}$ and
for almost all $n$, only nonzero negative-indexed coefficients
occur. Thus the $p_n(u)$ converge to $p_\infty(u)$ and thus
$p_\infty(u)$ has the same value as $p_n(u)$ while all its
coefficients are between $-2d$ and $+2d$. Thus all real numbers
can be represented in the given system.

For verifying that addition and comparison are $\omega$-automatic, one shows
that given $p = r_1+r_2-r_3$ computes by coefficientwise addition and
subtraction of the three numbers $r_1,r_2,r_3$ given in the normal form,
there is a deterministic B\"uchi automaton which recognises, on the
$\omega$-word of all $a_k$, the sign of
$p(u) = \sum_{k \in \mathbb Z} a_k u^k$ where
all sufficiently large $k$ have $a_k = 0$. The digits are bounded
by $6d$. The algorithm is like algorithm $C$ in Theorem~\ref{th:main},
when $k$ is the largest nonzero coefficient, the algorithm starts with
the memory $(a_{k+2},a_{k+1}) = (0,0)$ and then in each round updates
on reading $a_k$ the memory from $(a_{k+2},a_{k+1})$ to
$(a_{k+1}+2da_{k+2},a_k-a_{k+2})$, where the first number has at most
the absolute value $306d^2$ and the second number has at most
the absolute value $106d$. These values can only overshoot when
$a_{k+2}$ before the update has at least the absolute value $100d$.
So assume, without loss of generality, that $a_{k+2} > 100d$,
and that it is the first time where an overshooting will happen
as a consequence of an upcoming update. As the overshooting has
not yet happened, all $a_h$ with $h < k+2$ satisfy $|a_h| \leq 106d$.
Note that
$$100d > 106d \cdot \sum_{k < 0} u^k \geq
  106d \cdot (1/d+1/d^2+\ldots),$$
as $d>3$ and $1/d+1/d^2+\ldots \leq 1/3$.
Thus the algorithm can always indicate that the number is positive
in the case that an overshooting happens at an update where the
first component of the old memory pair is positive and the algorithm can
indicate that the number is negative in the case that an overshooting
happens at an update where the first component of the old memory pair
is negative. Furthermore, as there are only finitely many possibilities
for the two components in the memory, the algorithm works overall with
finite memory and can be implemented as a B\"uchi automaton which
goes into the right accepting state (for positive or negative) at
that moment that the sign of $p$ is known and which then stays
in this state forever.

While this allows directly to implement addition, subtraction and
comparisons, one has still to verify that multiplications with
finite polynomials $q(u)$ work. For the latter, note that a shift
by a costant amount of positions, say from $\sum_k a_k u^k$ to
$\sum_k a_k u^{k+h}$ for some constant $h$, is an $\omega$-automatic
function. Thus multiplication with fixed positive or negative powers
of $u$ is $\omega$-automatic.
Furthermore, repeated addition is $\omega$-automatic. For multiplication
with rationals involving $u$ like $x \mapsto (3+u)/(8-u^3) \cdot x$ can
be mapped back to multiplication with polynomials: In the theory
of $\omega$-automatic functions, one does not compute
the $y = (3+u)/(8-u^3) \cdot x$ directly but one compares whether $y$ is the
result of that operation. For that reason, one can just compare whether
$y \cdot (8-u^3) = x \cdot (3+u)$ and this can be implemeted via
multiplication with fixed polynomials in~$u$.~$\Box$
\end{proof}

\medskip
\noindent
This also works with the field of $\mathbb Q$ extended by $\sqrt[3]{7}$
or $\sqrt[3]{65}$
using the formulas and algorithms given in Example~\ref{3R7}. So one
has the following corollary.

\begin{cor}
For any semiautomatic subring $(G,+,<,=;\cdot)$ of the reals
constructed by methods in the present work,
there is an $\omega$-automatic representation of the reals where
addition, subtraction and comparisons are $\omega$-automatic and
furthermore the multiplication with any constant from $G$
is also an $\omega$-automatic unary function.
\end{cor}

\section{Conclusion}

\noindent
Jain, Khoussainov, Stephan, Teng and Zou \cite{JKSTZ17} studied
semiautomatic structures and in particular semiautomatic ordered
rings where addition, subtraction, order and equality are in fact
automatic. In particular they showed that for all nonsquare positive
integers $c$ the ordered ring $({\mathbb Z}[\sqrt{c}],+,<,=;\cdot)$
is semiautomatic. The present work shows that for certain roots one
can also add in the fraction $1/2$ so that the rings
$({\mathbb D}_2[\sqrt{c}],+,<,=;\cdot)$ with $c=2,3$, respectively,
are semiautomatic.

The case $c=3$ has applications in
geometry, as one can take the domain $G$ of the ring as a grid
in order to represent geometric objects with corner points in $G \times G$.
This grid $G$ in particular allows to represent equilateral triangles
and also allows to rotate objects by $30^\circ$ in the representation.
For the case $c=2$, one obtains a grid where one can rotate by $45^\circ$.
It is unknown whether one can make a grid which allows to rotate by $15^\circ$,
the challange would be that one has to get $\sqrt{2}$ and $\sqrt{3}$, $1/2$
all three into the grid in a way that addition, comparison, equality and
multiplication with fixed elements from the grid are automatic.

Further results are investigations into which collections of geometric objects
form an automatic family; while the collection of all triangles with corner
points in $G \times G$ is for no grid $G$ an automatic family,
all triangles and polygones with corner points in $G \times G$
are always regular. Axis-parallel rectangles form an automatic family
and for some grids, the relation saying that two axis-parallel rectangles
have the same size is semiautomatic.

Finally the paper provides two semiautomatic rings which contain
$\sqrt[3]{7}$ and $\sqrt[3]{65}$, respectively, and these methods
can generate further such examples. Some subsequent discussions transfer
the results to $\omega$-automatic structures.
For $\omega$-automatic structures, as for every real $r$ and every integer
$b$ the number $r/b$ always exists, one does not need,
as in the case of countable grids, to make sure that one can divide
by $b$ explicitly; it comes for free whenever addition is $\omega$-auto\-matic,
that is, if one has an $\omega$-automatic model of the reals with addition
and order, then the multiplication with every fixed rational is automatic;
so the main goal is to extend this also to multiplication with some
irrationals and this can be done for all square-roots of
positive integers by transferring results from prior work and also
for $\sqrt[3]{7}$ and $\sqrt[3]{65}$ by applying methods in the present work.
Note that enabling the multiplication of different irrationals might require
different representations and therefore the current methods do not give
a single $\omega$-automatic structure where the multiplication with each
algebraic number is $\omega$-automatic.
\fi


\begin{thebibliography}{99}

\ifx\versy\fully
\bibitem{Bl99} Achim Blumensath.
\newblock {\em Automatic structures}.
\newblock Diploma thesis, Department of Computer Science, RWTH Aachen, 1999.
\fi

\bibitem{BG00} Achim Blumensath and Erich Gr\"adel.
\newblock Automatic structures.
\newblock{\em Fifteenth Annual IEEE Symposium on Logic in Computer
Science}, LICS 2000, pages 51--62, 2000.

\bibitem{BJW05} Bernard Boigelot, Sebastien Jodogne and Pierre Wolper.
\newblock An effective decision procedure for linear arithmetic
  over the integers and reals.
\newblock {\em ACM Transactions on Computational Logic}, 6(3):614--633, 2005.

\bibitem{CJSS12} John Case, Sanjay Jain, Samuel Seah and Frank Stephan.
\newblock Automatic functions, linear time and learning.
\newblock {\em Logical Methods in Computer Science}, 9(3), 2013.

\bibitem{CK97} Karel Culik II and Jarkko Kari.
\newblock Computational fractal geometry with WFA.
\newblock {\em Acta Informatica}, 34:151--166, 1997.

\bibitem{ECHLPT92} David B.A.\ Epstein, James W.\ Cannon, Derek F.\ Holt,
  Silvio V.F.\ Levy, Micheal S.\ Paterson and William P.\ Thurston.
\newblock {\em Word Processing in Groups.}
\newblock Jones and Bartlett Publishers, Boston, 1992.

\bibitem{Ho76} Bernard R. Hodgson.
\newblock {\em Th\'eories d\'ecidables par automate fini}.
\newblock Ph.D.\ thesis, D\'epartement de
math\'ematiques et de statistique, Universit\'e de Montr\'eal, 1976.

\bibitem{Ho83} Bernard R. Hodgson.
\newblock D\'ecidabilit\'e par automate fini.
\newblock {\em Annales des sciences math\'ematiques du Qu\'ebec},
  7(1):39--57, 1983.

\bibitem{JKSTZ17} Sanjay Jain, Bakhadyr Khoussainov, Frank Stephan, Dan
Teng and Siyuan Zou.
\newblock Semiautomatic structures.
\newblock {\em Theory of Computing Systems}, 61(4):1254--1287, 2017.

\bibitem{JOPS10} Sanjay Jain, Yuh Shin Ong, Shi Pu and Frank Stephan.
\newblock On automatic families.
\newblock {\em Proceedings of the eleventh Asian Logic Conference}
  in honour of Professor Chong Chitat on his sixtieth birthday,
\newblock pages 94--113, World Scientific, 2012.

\bibitem{JSY07} Helmut J\"urgensen, Ludwig Staiger and Hideki Yamasaki.
\newblock Finite automata encoding geometric figures.
\newblock {\em Theoretical Computer Science}, 381(2--3):20--30, 2007.

\ifx\versy\fully
\bibitem{KKM11} Olga Kharlampovich, Bakhadyr Khoussainov and Alexei Miasnikov.
\newblock From automatic structures to automatic groups.
\newblock  {\em Groups, Geometry and Dynamical Systems}, 8(1):157--198, 2014.
\fi

\bibitem{KM10} Bakhadyr Khoussainov and Mia Minnes.
\newblock Three lectures on automatic structures.
\newblock {\em Logic Colloquium 2007}, Proceedings. 
  {\em Lecture Notes in Logic}, 35:132--176, 2010.

\bibitem{KN95} Bakhadyr Khoussainov and Anil Nerode.
\newblock Automatic presentations of structures.
\newblock {\em Logic and Computational Complexity, International Workshop},
  LCC 1994, Indianapolis, Indiana, USA, October 13--16, 1994, Proceedings.
  {\em Springer LNCS}, 960:367--392, 1995.

\bibitem{KNRS04} Bakhadyr Khoussainov, Andre Nies, Sasha Rubin and
  Frank Stephan.
\newblock Automatic structures: richness and limitations.
\newblock {\em Logical Methods in Computer Science}, 3(2), 2007.

\ifx\versy\fully
\bibitem{Ni07} Andr\'e Nies.
\newblock Describing Groups.
\newblock {\em The Bulletin of Symbolic Logic}, 13(3):305-339, 2007.
\fi

\bibitem{NS09} Andr\'e Nies and Pavel Semukhin.
\newblock Finite automata presentable Abelian groups.
\newblock {\em Annals of Pure and Applied Logic}, 161:458--467, 2009.

\bibitem{NT08} Andr\'e Nies and Richard Thomas.
\newblock FA-presentable groups and rings.
\newblock {\em Journal of Algebra}, 320:569-585, 2008.

\bibitem{OT05} Graham Oliver and Richard M.\ Thomas.
\newblock Automatic presentations for finitely generated groups.
\newblock {\em Twenty\-se\-cond Annual Symposium on Theoretical Aspects
  of Computer Science} (STACS 2005), Stuttgart, Germany, Proceedings.
  {\em Springer LNCS}, 3404:693--704, 2005.

\bibitem{Ru08} Sasha Rubin.
\newblock Automata presenting structures: a survey of the finite string case.
\newblock {\em The Bulletin of Symbolic Logic}, 14:169--209, 2008.

\bibitem{St15} Frank Stephan.
\newblock Automatic structures -- recent results and open questions.
\newblock Keynote Talk.
\newblock {\em Third International Conference on Science and Engineering
 in Mathematics, Chemistry and Physics}, ScieTech 2015,
\newblock {\em Journal of Physics: Conference Series}, 622/1
  (Paper 012013), 2015.

\bibitem{Ts11} Todor Tsankov.
\newblock The additive group of the rationals does not have an
  automatic presentation.
\newblock {\em The Journal of Symbolic Logic}, 76(4):1341--1351, 2011.

\ifx\versy\fully
\bibitem{ZGK12} Faried Abu Zaid, Erich Gr\"adel and Lukasz Kaiser.
\newblock The Field of Reals is not omega-automatic.
\newblock {\em Twenty-Ninth Symposium on Theoretical Aspects of
  Computer Science},
\newblock STACS 2012, pages 577--588, HAL archives-ouvertes, 2012.
\fi

\bibitem{ZGKP13} Faried Abu Zaid, Erich Gr\"adel, Lukasz Kaiser
  and Wied Pakusa.
\newblock Model-theoretic properties of $\omega$-automatic structures.
\newblock {\em Theory of Computing Sytems}, 55:856--880, 2014.

\end{thebibliography}
\end{document}